\documentclass[11pt,a4paper]{article}

\usepackage{amsfonts}
\usepackage{amsmath}
\usepackage{amssymb}
\usepackage{amsthm}
\usepackage{fullpage}
\usepackage{mleftright}
\usepackage{enumitem}
\usepackage{color}
\usepackage[normalem]{ulem}

\DeclareFontFamily{U}{mathx}{\hyphenchar\font45}
\DeclareFontShape{U}{mathx}{m}{n}{
	<5> <6> <7> <8> <9> <10>
	<10.95> <12> <14.4> <17.28> <20.74> <24.88>
	mathx10
}{}
\DeclareSymbolFont{mathx}{U}{mathx}{m}{n}
\DeclareFontSubstitution{U}{mathx}{m}{n}
\DeclareMathAccent{\widecheck}{0}{mathx}{"71}

\DeclareMathOperator*{\E}{\mathbb{E}}
\DeclareMathOperator*{\Var}{\mathrm{Var}}

\newcommand*{\norm}[1]{\left\|#1\right\|}

\newcommand*{\dist}{\mbox{dist}}

\newcommand{\R}{\mathbb{R}}

\newcommand{\F}{\mathbb{F}}
\newcommand*\cupdot{\mathbin{\mathaccent\cdot\cup}}

\newcommand*{\sparse}{\mbox{{\small SPARSE}}}
\newcommand*{\subsparse}{\mbox{{\scriptsize SPARSE}}}
\newcommand*{\dense}{\mbox{{\small DENSE}}}

\numberwithin{equation}{section}
\newtheorem{theorem}{Theorem}[section]
\newtheorem*{theorem*}{Theorem}
\newtheorem{lemma}[theorem]{Lemma}

\newtheorem{proposition}[theorem]{Proposition}
\newtheorem{corollary}[theorem]{Corollary}

\theoremstyle{definition}
\newtheorem{definition}[theorem]{Definition}
\newtheorem*{remark*}{Remark}

\mleftright
\begin{document}

\title{Double Balanced Sets in High Dimensional Expanders\footnote{Originally appeared in RANDOM 2022~\cite{KM22}.}}
\author{
	Tali Kaufman\thanks{Bar-Ilan University, Israel. Email: \texttt{kaufmant@mit.edu}. Supported by ERC and BSF.} \and
	David Mass\thanks{Bar-Ilan University, Israel. Email: \texttt{dudimass@gmail.com}. Supported by the Adams Fellowship Program of the Israel Academy of Sciences and Humanities.}}
\maketitle

\begin{abstract}
Recent works have shown that expansion of pseudorandom sets is of great importance. %As an example, the discovery that pseudorandom small sets in the Grassmann graph expand near perfectly has led to the resolution of Khot's 2-to-2 Games Conjecture.
However, all current works on pseudorandom sets are limited only to product (or approximate product) spaces, where Fourier Analysis methods could be applied. In this work we ask the natural question whether pseudorandom sets are relevant in domains where Fourier Analysis methods cannot be applied, e.g., one-sided local spectral expanders.

We take the first step in the path of answering this question. We put forward a new definition for pseudorandom sets, which we call ``double balanced sets''.
We demonstrate the strength of our new definition by showing that small double balanced sets in one-sided local spectral expanders have very strong expansion properties, such as unique-neighbor-like expansion. We further show that cohomologies in cosystolic expanders are double balanced, and use the newly derived strong expansion properties of double balanced sets in order to obtain an exponential improvement over the current state of the art lower bound on their minimal distance.
\end{abstract}

\section{Introduction}
The study of pseudorandom (or ``global'') functions has led to many recent advancements. It has been shown that they possess an effective hypercontractive inequality in many domains such as the $p$-biased cube~\cite{KLLM19}, the slice~\cite{KMMS18}, the Grassmann graph~\cite{KMS18} and two-sided local spectral expanders~\cite{BHKL21}. The common observation in all of these works is that while hypercontractivity does not hold for any general function, it holds for a certain subclass of pseudorandom functions. This phenomenon has been the key to many breakthroughs, most famously the resolution of Khot's 2-to-2 Games Conjecture~\cite{KMS18}.

While this study of pseudorandom functions has been very fruitful in many domains, currently it is still limited only to domains where Fourier Analysis methods could be applied. These domains are product (or approximate product) spaces, so each function has an orthogonal (or an approximate orthogonal) decomposition. While these domains are enough for a lot of applications, there are many applications that require other domains. Some examples are the recent works on efficient sampling algorithms (e.g.,~\cite{ALOV19,ALO20,AL20,AJK+21} and more). The domains in these works are one-sided local spectral expanders, which inherently do not possess an orthogonal decomposition.

In this work we make the first step in the study of pseudorandom functions in other domains where Fourier Analysis methods cannot be applied. We put forward an alternative definition for pseudorandom functions, which we call ``double balanced sets''. We demonstrate the strength of our new definition by showing that small double balanced sets in one-sided local spectral expanders have very strong expansion properties. We further show that cohomologies in cosystolic expanders are double balanced, and then by the strong expansion properties of double balanced sets, we achieve an exponential improvement over the state of the art lower bound on their minimal distance.

\subsection{Double balanced sets}
In order to present our definition of double balanced sets, we need to set some notations first. A $d$-dimensional simplicial complex $X$ is a $(d+1)$-hypergraph which is closed under inclusions, i.e., if $\sigma \in X$ then every $\tau \subseteq \sigma$ is also in $X$. A $k$-face is a hyperedge of size $k+1$ and the set of $k$-faces in the complex is denoted by $X(k)$. For any face $\sigma \in X$, the link of $\sigma$, denoted by $X_\sigma$, is the subcomplex that is obtained by all the faces that contain $\sigma$ and then removing $\sigma$ from all of them.

Let $f \subseteq X(k)$ be a subset of $k$-faces in $X$. For any face $\sigma \in X(\ell)$, $\ell < k$, we denote by $f_\sigma \subseteq X_\sigma(k-\ell-1)$ the \emph{localization} of $f$ to the link of $\sigma$, where a face $\tau \in X_\sigma(k-\ell-1)$ is in $f_\sigma$ if and only if $\tau \cup \sigma \in f$. We also denote by $f^\sigma$ the \emph{restriction} of $f$ to the link of $\sigma$, where $f^\sigma = f \cap X_\sigma(k)$. Note that both $f_\sigma$ and $f^\sigma$ ``live'' in the link of $\sigma$, but $f_\sigma$ is a subset of $(k-\ell-1)$-faces whereas $f^\sigma$ is a subset of $k$-faces.

For simplicity, we assume in the introduction that the complex has a uniform probability distribution in every dimension. In the body of the paper we will take into account general probability distributions.

\begin{definition}[Double balanced sets]
We say that $f \subseteq X(k)$ is \emph{$\alpha$-double balanced in dimension $\ell$}, $\ell < k$, if for every $\ell$-face $\sigma \in X(\ell)$ it holds that
\begin{equation}
\frac{|f_\sigma|}{|X_\sigma(k-\ell-1)|} \le
\alpha\E_{v \in \sigma}\left[\frac{\left|(f_{\sigma\setminus v})^v\right|}{|X_\sigma(k-\ell)|}\right].
\end{equation}\label{intro-eq:double-balanced}
We say that $f$ is \emph{$\alpha$-double balanced} if it is $\alpha$-double balanced in dimension $\ell$ for every $\ell < k$.
\end{definition}

In order to get some intuition, let us focus on low dimensions first. Let $X$ be a $3$-dimensional complex and $f \subseteq X(2)$ (i.e., a set of triangles in a complex with pyramids).
\begin{itemize}
\item
For every vertex $v \in X(0)$, the left-hand side of~\eqref{intro-eq:double-balanced} translates to the fraction of triangles in $f$ that contain $v$ out of all the triangles that contain $v$, and the right-hand side of~\eqref{intro-eq:double-balanced} translates to $\alpha$ times the fraction of triangles in $f$ that together with $v$ form a pyramid out of all the pyramids that contain $v$.

\item For every edge $\{u,v\} \in X(1)$, the left-hand side of~\eqref{intro-eq:double-balanced} translates to the fraction of triangles in $f$ that contain $\{u,v\}$ out of all the triangles that contain $\{u,v\}$, and the right-hand side of~\eqref{intro-eq:double-balanced} translates to $\alpha$ times the average fraction of triangles in $f$ that contain $u$ or $v$ and together with $v$ or $u$, respectively, form a pyramid out of all the pyramids that contain $\{u,v\}$.
\end{itemize}

In general, the left-hand side of~\eqref{intro-eq:double-balanced} translates to the fraction of $k$-faces in $f$ that contain $\sigma$, and the right-hand side translates to the average fraction of $k$-faces in $f$ that contain $|\sigma|-1$ vertices from $\sigma$ and together with $\sigma$ forms a $(k+1)$-face.

Let us explain briefly the motivation behind this definition. From a spectral point of view, it is known that high dimensional random walks with intersections do not mix rapidly, whereas random walks without intersections (also known as swap walks~\cite{AJT19} or complement walks~\cite{DD19}) have an optimal mixing rate\footnote{By random walks with intersections we mean that we move from an $i$-face $\sigma$ to a $j$-face $\tau$ through a $k$-face that contain both $\sigma$ and $\tau$, where the intersection $\sigma \cap \tau$ may be non-empty, whereas random walks without intersections require that $\sigma \cap \tau$ would be empty.}. Previous works on pseudorandom sets (e.g., \cite{BHKL21}) benefit from the optimal mixing rate of non-intersecting random walks, but for that the complex has to be of a very high dimension, i.e., in order to gain anything on a pseudorandom set of dimension $k$, the complex has to be of dimension at least $2k$ (so we can move between $k$-faces without intersections). Our definition of double balanced sets benefits from the optimal mixing rate of non-intersecting random walks \emph{even when $d = k+1$}. The reason is that the right-hand side of~\eqref{intro-eq:double-balanced}, when viewed in the link of $\sigma \setminus v$ for some vertex $v \in \sigma$, is concerned with faces that do not contain $v$, i.e., it is related to a non-intersecting random walk inside the link of $\sigma \setminus v$.

From a topological point of view, our definition of double balanced sets relates faces of two consecutive dimensions (i.e., $(k-\ell-1)$-faces in the left-hand side of~\eqref{intro-eq:double-balanced} and $(k-\ell)$-faces in the right-hand side), similar to usual topological operators (e.g., the boundary and coboundary operators). In this sense, our definition has the potential to benefit also from the topological properties of the complex. Indeed, we show that cohomologies in high dimensional expanders are double balanced by utilizing the topological expansion of the complex.

To summarize the above discussion, our definition of double balanced sets has the potential to imitate a situation where the complex has many dimensions above (like in previous works) while having only one dimension above. It benefits both from spectral and topological properties of the complex, whereas previous works could only use spectral properties. We believe that utilizing the topological properties of the complex, as well as spectral properties, would lead to many breakthroughs in the future.

\subsection{Relation to the common definition}
We would like to formalize the intuitive similarity of our new definition (of double balanced sets) to the common definition (of pseudorandom sets).

The common definition of pseudorandom sets, as given in~\cite{BHKL21}\footnote{The actual definition is considered with general functions from $X(k)$ to $\mathbb{R}$. For simplicity we consider only functions from $X(k)$ to $\{0,1\}$, i.e., functions that correspond to subsets of $k$-faces.}, says that a set of $k$-faces $f$ is $\varepsilon$-pseudorandom in dimension $\ell$, $\ell < k$, if for every $\ell$-face $\sigma \in X(\ell)$ it holds that
\begin{equation}
\frac{|f_\sigma|}{|X_\sigma(k-\ell-1)|} \le \varepsilon.
\end{equation}

As demonstrated in the following lemma, our definition of double balanced sets implies almost pseudorandomness.
\begin{lemma}\label{intro-lem:relation-to-common-def}
Let $X$ be a good enough one-sided local spectral expander\footnote{The definition of one-sided local spectral expansion will be introduced later in the paper.}. For any $\alpha$-double balanced set of $k$-faces $f \in X(k)$ and any dimension $\ell < k$, if
$$\frac{|f|}{|X(k)|} \le \frac{\varepsilon}{(\ell+1)\alpha^\ell}$$
then
$$\Pr_{\sigma \in X(\ell)}\left[\frac{|f_\sigma|}{|X_\sigma(k-\ell-1)|} \le \varepsilon\right] \ge 1-\varepsilon\frac{|f|}{|X(k)|}.$$
\end{lemma}

In words, for a sufficiently small set, if the set is $\alpha$-double balanced then it is also almost pseudorandom, i.e., all $\ell$-faces besides of a negligible fraction of them satisfy the pseudorandomness property.

\subsection{Inheritance property}
An interesting property that applies to double balanced sets is that it is inherited by lower dimensions. We show that a set of $k$-faces which is double balanced in dimension $\ell$ is also double balanced in all dimensions below $\ell$. This result is obtained by applying the following lemma step by step.

\begin{lemma}[Double balance inheritance]\label{intro-lem:double-balanced-inheritance}
If $f \subseteq X(k)$ is $\alpha$-double balanced in dimension $\ell$, then $f$ is $\alpha'$-double balanced in dimension $\ell-1$, where
$$\alpha' = \frac{\alpha\ell}{\ell+1-\alpha}.$$
\end{lemma}

It is worth to note that when $f$ is perfectly double balanced, i.e., when $\alpha=1$, then lemma~\ref{intro-lem:double-balanced-inheritance} implies that $f$ is also perfectly double balanced in all dimensions below $\ell$. In other words, perfect double balance is inherited by lower dimensions \emph{without any loss}.

\subsection{$\delta_1$-expansion of small double balanced sets}
In recent years, a few different notions of high dimensional expansion have been studied. One such notion is $\delta_1$-expansion, which can be viewed as a generalization of unique-neighbor expansion in graphs. It is a strong expansion notion that is usually very hard to get. For a set of $k$-faces $f \subseteq X(k)$, $\delta_1(f)$ is defined as the set of $(k+1)$-faces that contain exactly one $k$-face from $f$. We say that $f$ is $\delta_1$-expanding if
\begin{equation}\label{intro-eq:delta1-expansion}
\frac{|\delta_1(f)|}{|X(k+1)|} \ge \varepsilon\frac{|f|}{|X(k)|}.
\end{equation}
In~\cite{KM21} it has been shown that $\delta_1$-expansion for small sets implies group-independent cosystolic expansion, i.e., cosystolic expansion over any group.

In order to demonstrate the strength of our definition of double balanced sets, we show that small double balanced sets are $\delta_1$-expanding. On one hand, we show that when a double balanced set $f$ is sufficiently small, it has a \emph{nearly perfect} $\delta_1$-expansion, i.e., $\varepsilon$ in equation~\eqref{intro-eq:delta1-expansion} is very close to $k+2$. On the other hand, for larger double balanced sets (which are still small, but not that small), we show that they have some $\delta_1$-expansion, i.e., $\varepsilon > 0$ in equation~\eqref{intro-eq:delta1-expansion}. We prove the following two theorems. 

\begin{theorem}[Nearly optimal $\delta_1$-expansion for sufficiently small double balanced sets]\label{intro-thm:optimal-delta1-expansion}
Let $X$ be a good enough one-sided local spectral expander.
For any $\alpha$-double balanced set of $k$-faces $f \subseteq X(k)$ and $\varepsilon > 0$, if
$$\frac{|f|}{|X(k)|} \le \frac{\varepsilon}{(k+1)^2\alpha^k}$$
then
$$\frac{|\delta_1(f)|}{|X(k+1)|} \ge (1-3\varepsilon)(k+2)\frac{|f|}{|X(k)|}.$$
\end{theorem}

\begin{theorem}[Some $\delta_1$-expansion for small double balanced sets]\label{intro-thm:some-delta1-expansion}
Let $X$ be a good enough one-sided local spectral expander.
For any $\alpha$-double balanced set of $k$-faces $f \subseteq X(k)$ and $\varepsilon > 0$, if
$$\frac{|f|}{|X(k)|} \le \frac{1-\varepsilon}{(k+1)\alpha^k}$$
then
$$\frac{|\delta_1(f)|}{|X(k+1)|} > 0.$$
\end{theorem}

Both of theorems~\ref{intro-thm:optimal-delta1-expansion} and~\ref{intro-thm:some-delta1-expansion} demonstrate the strength of our definition of double balanced sets. The key idea that since $f$ is a small set, its double balance property implies that it has to be small in every link as well, which in turn implies $\delta_1$-expansion. The novelty over previous works (e.g., \cite{KKL14,EK16,KM21}) is to benefit from the optimal mixing rate of non-intersecting random walks. As explained in section 1.1, our definition of double balanced sets is related in a sense to non-intersecting random walks and hence benefits from an optimal mixing rate. This is in contrast to previous works, which essentially used only intersecting random walks, and hence could obtain worse bounds and only for much smaller sets.

\subsection{Application to minimal distance of cohomologies}
Cohomologies stand in the center of recent studies in Mathematics, and they have already found some applications in Theoretical Computer Science as well. Complexes with large cohomologies have played a key role in the construction of efficiently decodable quantum LDPC codes with a large distance~\cite{EKZ20}. It is known by now to construct quantum LDPC codes with a larger distance~\cite{PK21,LZ22}, however these are not known to be efficiently decodable. Complexes with large cohomologies were also the main block in the first construction of explicit 3XOR instances that are hard for the Sum-of-Squares Hierarchy~\cite{DFHT20}. Other constructions which are hard for more levels of the the Sum-of-Squares Hierarchy~\cite{HL22} are known by now. Nonetheless, the construction of~\cite{DFHT20} is still the best known construction from \emph{simplicial} complexes and it has been the first step in this line of works.

In order to define cohomologies, let us identify a set of $k$-faces in $X$ with an $\F_2$-valued function $f : X(k) \to \F_2$ and denote by $C^k(X)$ the space of all $\F_2$-valued functions on $X(k)$. The coboundary operator $\delta^k : C^k(X) \to C^{k+1}(X)$ is defined by
$$\delta^k f(\sigma) = \sum_{u \in \sigma}f(\sigma \setminus \{u\}) \mbox{ mod } 2.$$
The image of $\delta^{k-1}$ is called the $k$-coboundaries and is denoted by
$$B^k(X) = \{\delta^{k-1} f \;|\; f \in C^{k-1}(X)\}.$$
The kernel of $\delta^k$ is called the $k$-cocycles and is denoted by
$$Z^k(X) = \{f \in C^k(X) \;|\; \delta^k f = \mathbf{0}\}.$$
It is not hard to check that $B^k(X) \subseteq Z^k(X) \subseteq C^k(X)$.
The $k$-cohomology of $X$ is the quotient space $H^k(X) = Z^k(X)/B^k(X)$.

Previous works could only obtain complexes with some constant lower bound on the size of their cohomologies~\cite{KKL14,EK16,KM21}. We show that for high dimensional expanders (in a topological sense), all of their cohomology elements are double balanced. We then utilize the $\delta_1$-expansion of double balanced sets in order to obtain a lower bound on their size, achieving an \emph{exponential} improvement upon the current state of the art.

\begin{theorem}[Cohomologies are double balanced]\label{intro-thm:cohomologies-double-balanced}
For a complex whose links are topological expanders, every $k$-cohomology element is $((k+1)/\beta$)-double balanced, where $\beta$ is the expansion constant in the links of the complex.
\end{theorem}

\begin{theorem}[Lower bound on cohomology elements]\label{intro-thm:lower-bound-cohomologies}
For a good enough one-sided local spectral expander whose links are topological expanders, every $k$-cohomology element must be of density at least $\beta^k/(k+1)!$, where $\beta$ is the expansion constant in the links of the complex.
\end{theorem}

\begin{remark*}
The current state of the art lower bound on the size of cohomologies prior to this work is
$\approx (\beta^k/k!)^{2^k}$ \cite[Lemma~3.10]{KM21}.
\end{remark*}

\subsection{Organization}
In section 2 we provide the required preliminaries. In section 3 we introduce the formal definition of double balanced sets and prove its inheritance property. In section 4 we show that small double balanced sets in one-sided local spectral expanders have the strong $\delta_1$-expansion property, and also explain how to prove lemma~\ref{intro-lem:relation-to-common-def}. In section 5 we show that cohomologies in a complex with topological expanding links are double balanced, obtaining an exponential improvement upon the current state of the art lower bound on their minimal distance.

\section{Preliminaries}

\subsection{Simplicial complexes}
Recall that a $d$-dimensional simplicial complex $X$ is a downwards closed $(d+1)$-hypergraph. A $k$-face of $X$ is a hyperedge of size $k+1$, and the set of $k$-faces of $X$ is denoted by $X(k)$. An assignment of values from $\F_2$ to the $k$-faces, $k\le d$, is called a $k$-cochain, and the space of all $k$-cochains over $\F_2$ is denoted by $C^k(X)$.

Any assignment to the $k$-faces $f \in C^k(X)$ induces an assignment to the $(k+1)$-faces by the coboundary operator $\delta$. For any $(k+1)$-face $\sigma = \{v_0,\dotsc,v_{k+1}\}$, $\delta(f)(\sigma)$ is defined by
$$\delta(f)(\sigma) = \sum_{i=0}^{k+1}f(\sigma \setminus \{v_i\})\quad (\mbox{mod } 2).$$

The kernel of the coboundary operator is called the \emph{$k$-cocycles} and denoted by
$$Z^k(X) = \{f \in C^k(X) \;|\; \delta(f) = \mathbf{0} \}.$$
The image of $\delta$ is called the $k$-coboundaries and denote by
$$B^k(X) = \{\delta(f) \;|\; f \in C^{k-1}(X) \}.$$
One can check that $\delta(\delta(f)) = \mathbf{0}$ always holds, hence $B^k(X) \subseteq Z^k(X) \subseteq C^k(X)$. The quotient space $Z^k(X)/B^k(X)$ is called the $k$-cohomologies and denoted by $H^k(X)$.

For a $d$-dimensional simplicial complex $X$, let $P_d:X(d) \to \R_{\ge 0}$ be a probability distribution over the $d$-faces of the complex. For simplicity, we will assume in this work that $P_d$ is the uniform distribution. This probability distribution over the $d$-faces induces a probability distribution $P_k$ for every dimension $k < d$ by selecting a $d$-face $\sigma_d$ according to $P_d$ and then selecting a $k$-face $\sigma_k \subset \sigma_d$ uniformly at random.

The weight of any $k$-cochain $f \in C^k(X)$ is defined by
$$\norm{f} = \Pr_{\sigma \sim P_k}[f(\sigma) \ne 0],$$
i.e., the (weighted) fraction of non-zero elements in $f$. The distance between two $k$-cochains $f,g \in C^k(X)$ is defined as $\dist(f,g) = \norm{f - g}$.

We also add a useful definition of a mutual weight of two cochains. For $\ell < k$ and two cochains $f \in C^k(X)$, $g \in C^\ell(X)$ we define their mutual weight by
$$\norm{(f,g)} = \Pr_{\sigma_k \sum P_k, \sigma_\ell \subset \sigma_k}[f(\sigma) \ne 0 \wedge \sigma_\ell \ne 0],$$
where $\sigma_k$ is chosen according to the distribution $P_k$ and $\sigma_\ell$ is an $\ell$-face chosen uniformly from $\sigma_k$ (i.e., $\sigma_\ell$ is chosen according to $P_\ell$ conditioned on $\sigma_k$ being chosen).

\subsection{Cosystolic and coboundary expansion}
Coboundary expansion has been introduced by Linial and Meshulam~\cite{LM06} and independently by Gromov~\cite{Gro10}. It is a generalization of edge expansion of graphs to higher dimensions.

\begin{definition}[Coboundary expansion]
	A $d$-dimensional simplicial complex $X$ is said to be an $\varepsilon$-coboundary expander if for every $k < d$ and $f \in C^k(X)\setminus B^k(X)$ it holds that
	$$\frac{\norm{\delta(f)}}{\dist(f,B^k(X))} \ge \varepsilon,$$
	where $\dist(f,B^k(X)) = \min\{\dist(f,g) \;|\; g \in  B^k(X) \}$.
\end{definition}

Cosystolic expansion is similar to coboundary expansion, with the main difference that it can have non-trivial cohomologies as long as they are large.

\begin{definition}[Cosystolic expansion]
	A $d$-dimensional simplicial complex $X$ is said to be an $(\varepsilon,\mu)$-cosystolic expander if for every $k < d$:
	\begin{enumerate}
		\item For any $f \in C^k(X)\setminus Z^k(X)$ it holds that
		$$\frac{\norm{\delta(f)}}{\dist(f,Z^k(X))} \ge \varepsilon,$$
		where $\dist(f,Z^k(X)) = \min\{\dist(f,g) \;|\; g \in  Z^k(X) \}$.
		\item For any $f \in Z^k(X)\setminus B^k(X)$ it holds that $\norm{f} \ge \mu$.
	\end{enumerate}
\end{definition}

\subsection{Links, localization and restriction}
For every face $\sigma \in X$, its local view, also called its \emph{link}, is a $(d-|\sigma|-1)$-dimensional simplicial complex defined by $X_\sigma = \{\tau \setminus \sigma \;|\; \sigma \subseteq \tau \in X\}$. The probability distribution over the top faces of $X_\sigma$ is induced from the probability distribution of $X$, where for any top face $\tau \in X_\sigma(d-|\sigma|-1)$, its probability is the probability to choose $\sigma \cup \tau$ in $X$ conditioned on choosing $\sigma$. Since we assume in this work that the probability distribution over the top faces of $X$ is the uniform distribution, it follows that the probability distribution over the top faces of $X_\sigma$ is the uniform distribution.

For any $k$-cochain $f \in C^k(X)$ and an $\ell$-face $\sigma \in X(\ell)$, the \emph{localization} of $f$ to the link of $\sigma$ is a $(k-\ell-1)$-cochain in the link of $\sigma$, $f_\sigma \in C^{k-\ell-1}(X_\sigma)$ defined by $$f_{\sigma}(\tau) = f(\sigma\cup\tau).$$ The \emph{restriction} of $f$ to the link of $\sigma$ is a $k$-cochain in the link of $\sigma$, $f^\sigma \in C^k(X_\sigma)$ defined by $$f^\sigma(\tau) = f(\tau).$$

\subsection{Local spectral expansion}
Another notion of high dimensional expansion, called \emph{local spectral expansion} is concerned with the spectral properties of the links of the complex.

\begin{definition}[Two-sided local spectral expansion]
	A $d$-dimensional simplicial complex $X$ is called a \emph{$\lambda$-two-sided local spectral expander}, $\lambda > 0$, if for every $k \le d-2$ and $\sigma \in X(k)$, the underlying graph\footnote{The graph whose vertices are $X_\sigma(0)$ and its edges are $X_\sigma(1)$.} of $X_\sigma$ is a $\lambda$-two-sided spectral expander, i.e., its spectrum is bounded from above by $\lambda$ and from below by $-\lambda$.% satisfies $\lambda_2 \le \lambda$ and $\lambda_n \ge -\lambda$, where $\lambda_1 \ge \lambda_2 \ge \dotsb \ge \lambda_n$ are its eigenvalues.
\end{definition}

\begin{definition}[One-sided local spectral expansion]
	A $d$-dimensional simplicial complex $X$ is called a \emph{$\lambda$-one-sided local spectral expander}, $\lambda > 0$, if for every $k \le d-2$ and $\sigma \in X(k)$, the underlying graph$^4$ of $X_\sigma$ is a $\lambda$-one-sided spectral expander, i.e., its spectrum is bounded from above by $\lambda$.% satisfies $\lambda_2 \le \lambda$, where $\lambda_1 \ge \lambda_2 \ge \dotsb \ge \lambda_n$ are its eigenvalues.
\end{definition}

\subsection{Minimal and locally minimal cochains}
One of the technical notions we use in this work is the notion of a minimal cochain. We say that a $k$-cochain $f \in C^k(X)$ is \emph{minimal} if its weight cannot be reduced by adding a coboundary to it, i.e., for every $g \in B^k(X)$ it holds that $\norm{f} \le \norm{f-g}$. Recall that the distance of $f$ from the coboundaries is defined by
$\dist(f, B^k(X)) = \min \{\norm{f-g} \;|\; g \in B^k(X) \}$. Since $\mathbf{0} \in B^k(X)$, it follows that for every $f \in C^k(X)$, $\norm{f} \ge \dist(f,B^k(X))$. Thus, $f$ is said to be \emph{minimal} if and only if $\norm{f} = \dist(f, B^k(X))$.

We also define the notion of a locally minimal cochain, where we say that $f \in C^k(X)$ is \emph{locally minimal} if for every vertex $v$, the localization of $f$ to the link of $v$ is minimal in the link, i.e., $f_v$ is minimal in $X_v$ for every $v \in X(0)$. It is not hard to check that any minimal cochain is also locally minimal.

\section{Double balanced sets}
We start by providing the formal definition of a double balanced cochain. Recall that for any $k$-cochain $f \in C^k(X)$ and a vertex $u \in X(0)$, we denote by $f^u$ the restriction of $f$ to the $k$-faces in the link of $u$, i.e., $f^u \in C^k(X_u)$.

\begin{definition}[Double balanced cochains]
	Let $X$ be a $d$-dimensional simplicial complex. A $k$-cochain $f \in C^k(X)$ is said to be \emph{$\alpha$-double balanced in dimension $\ell$}, where $\alpha \ge 1$ and $0 \le \ell \le k-1$, if for every $\ell$-face $\sigma \in X(\ell)$ it holds that
	$$\norm{f_\sigma} \le \alpha\E_{u \in \sigma}\norm{(f_{\sigma\setminus u})^u}.$$
	$f$ is said to be \emph{$\alpha$-double balanced} if $f$ is $\alpha$-double balanced in dimension $\ell$ for every $\ell < k$.
\end{definition}

\subsection{Balance inheritance}
An interesting property that applies to double balanced cochains is that it is inherited by lower dimensions. We show that a cochain of $k$-faces which is double balanced in dimension $\ell$ is also double balanced in all dimensions below $\ell$. We prove lemma~\ref{intro-lem:double-balanced-inheritance} from the introduction, which we restate here for convenience.

\begin{lemma}[Double balance inheritance]\label{lem:double-balanced-inheritance}
Let $f \in C^k(X)$ be an $\alpha$-double balanced cochain in dimension $\ell$. Then $f$ is $\alpha'$-double balanced in dimension $\ell-1$, where
$$\alpha' = \frac{\alpha\ell}{\ell+1-\alpha}.$$
\end{lemma}

\begin{proof}
	Let $\tau \in X(\ell-1)$.
%	\begin{align*}
%		\norm{f_\tau} &=
%		\E_{u \in X_\tau(0)}[\norm{f_{\tau u}}] \\&\le
%		\E_{u \in X_\tau(0)}\left[\alpha\E_{v \in \tau u}\big[\norm{(f_{\tau u \setminus v})^v}\big]\right] \\&=
%		\E_{u \in X_\tau(0)}\left[\frac{\alpha}{\ell+1}\sum_{v \in \tau u}\norm{(f_{\tau u \setminus v})^v}\right] \\&=
%		\frac{\alpha}{\ell+1}\E_{u \in X_\tau(0)}\left[\norm{(f_\tau)^u} + \sum_{v \in \tau}\norm{(f_{\tau u \setminus v})^v}\right] \\&=
%		\frac{\alpha}{\ell+1}\left(\E_{u \in X_\tau(0)}[\norm{(f_\tau)^u}] + \sum_{v \in \tau}\E_{u \in X_\tau(0)}[\norm{(f_{\tau u \setminus v})^v}]\right) \\&=
%		\frac{\alpha}{\ell+1}\left(\norm{f_\tau} + \sum_{v \in \tau}\norm{(f_{\tau \setminus v})^v}\right)
%	\end{align*}
	
	\begin{align*}
		\norm{f_\tau} &=
		\E_{u \in X_\tau(0)}[\norm{f_{\tau u}}] \\[5pt]&\le
		\E_{u \in X_\tau(0)}\left[\alpha\E_{v \in \tau u}\big[\norm{(f_{\tau u \setminus v})^v}\big]\right] \\[5pt]&=
		\E_{u \in X_\tau(0)}\left[\frac{\alpha}{\ell+1}\norm{(f_\tau)^u} + \frac{\alpha\ell}{\ell+1}\E_{v \in \tau}\big[\norm{(f_{\tau u \setminus v})^v}\big]\right] \\[5pt]&=
		\frac{\alpha}{\ell+1}\E_{u \in X_\tau(0)}[\norm{(f_\tau)^u}] + \frac{\alpha\ell}{\ell+1}\E_{v \in \tau}\left[\E_{u \in X_\tau(0)}\big[\norm{(f_{\tau u \setminus v})^v}\big]\right] \\[5pt]&=
		\frac{\alpha}{\ell+1}\norm{f_\tau} + \frac{\alpha\ell}{\ell+1}\E_{v \in \tau}\norm{(f_{\tau \setminus v})^v},
	\end{align*}
where the inequality follows since $f$ is $\alpha$-double balanced in dimension $\ell$ and all the other steps follow from laws of probability.
	This implies that
	\begin{align*}
		\norm{f_\tau} \le
		\frac{\alpha\ell}{\ell + 1 - \alpha}\E_{v \in \tau}\norm{(f_{\tau \setminus v})^v}.
	\end{align*}
\end{proof}

It is worth to note that when $f$ is perfectly double balanced, i.e., when $\alpha=1$, then lemma~\ref{lem:double-balanced-inheritance} implies that $f$ is also perfectly double balanced in all dimensions below $\ell$. In other words, perfect double balance is inherited by lower dimensions \emph{without any loss}.

\begin{corollary}
	Let $f \in C^k(X)$ be a $1$-double balanced cochain in dimension $\ell$. Then $f$ is also $1$-double balanced in all dimensions below $\ell$.
\end{corollary}

%It might be easier to understand the above definition by considering first a cochain on edges or triangles.
%	
%\begin{itemize}
%	\item A $1$-cochain (on edges) $f \in C^1(X)$ is said to be $\alpha$-double balanced in dimension $0$ if for every vertex $u \in X(0)$ it holds that $\norm{f_u} \le \alpha\norm{f^u}$.
%	\item A $2$-cochain (on triangles) $f \in C^2(X)$ is said to be $\alpha$-double balanced in dimension $0$ if for every vertex $u \in X(0)$ it holds that $\norm{f_u} \le \alpha\norm{f^u}$, and it said to be $\alpha$-double balanced in dimension $1$ if for every edge $\{u,v\} \in X(1)$ it holds that $\displaystyle\norm{f_{uv}} \le \frac{\alpha}{2}(\norm{(f_u)^v} + \norm{(f_v)^u})$.
%\end{itemize}

\section{$\delta_1$-expansion for small double balanced sets}
In this section we show that small double balanced sets are $\delta_1$-expanding. On one hand, we show that when a double balanced set $f$ is sufficiently small, it has a \emph{nearly optimal} $\delta_1$-expansion. On the other hand, for larger double balanced sets (which are still small, but not that small), we show that they have some $\delta_1$-expansion, i.e., $\norm{\delta_1(f)} > 0$. We prove theorems~\ref{intro-thm:optimal-delta1-expansion} and~\ref{intro-thm:some-delta1-expansion} from the introduction, which we restate here in a formal way. 

\begin{theorem}[Nearly optimal $\delta_1$-expansion for sufficiently small double balanced sets]\label{thm:optimal-delta1-for-small-sets}
	For every $d \ge 2$, $\alpha \ge 1$ and $0 < \varepsilon < 1$ there exists $\lambda = \lambda(d, \alpha, \varepsilon)$ such that the following holds:
	Let $X$ be a $d$-dimensional $\lambda$-one-sided local spectral expander. For any $k$-cochain $f \in C^k(X)$, $1 \le k < d$, such that $f$ is $\alpha$-double balanced and $\displaystyle \norm{f} \le \frac{\varepsilon}{(k+1)^2\alpha^k}$ it holds that
	$$\norm{\delta_1(f)} \ge (k+2)(1-3\varepsilon)\norm{f}.$$
	%$$\lambda^2 \le \frac{\varepsilon^3}{(d-1)d^4(d!)^3\alpha^{2(d-1)}}.$$
\end{theorem}

\begin{theorem}[Some $\delta_1$-expansion for small double balanced sets]\label{thm:some-delta1-for-small-sets}
	For every $d \ge 2$, $\alpha \ge 1$ and $0 < \varepsilon < 1$ there exists $\lambda = \lambda(d, \alpha, \varepsilon)$ such that the following holds:
	Let $X$ be a $d$-dimensional $\lambda$-one-sided local spectral expander. For any $k$-cochain $f \in C^k(X)$, $1 \le k < d$, such that $f$ is $\alpha$-double balanced and $\displaystyle \norm{f} \le \frac{1-\varepsilon}{(k+1)\alpha^k}$ it holds that
	$$\norm{\delta_1(f)} > 0.$$
	%$$\lambda^2 \le \frac{\varepsilon^3}{(d-1)d^4(d!)^3\alpha^{2(d-1)}}.$$
\end{theorem}

We split the proof of these theorems to two parts. In the first part we show that if almost all of the $(k-1)$-faces of a cochain are not dense then its $\delta_1$ is optimal. In the second part, we show that for sufficiently small double balanced cochains, almost all of their $(k-1)$-faces are indeed not dense.

\subsection{Part I - Bound $\delta_1(f)$ by the dense $(k-1)$-faces.}
Let $X$ be a $d$-dimensional $\lambda$-one-sided local spectral expander and $0 < \eta < 1$ a density constant.

For any $k$-cochain $f \in C^k(X)$ we define the set of dense $(k-1)$-faces by
$$\dense_{k-1} = \{\sigma \in X(k-1) \;|\; \norm{f_\sigma} > \eta\}.$$

We show in this section that $\norm{\delta_1(f)}$ can be bounded by the fraction of dense $(k-1)$-faces.%, or in other words, if the fraction of dense $(k-1)$-faces is small enough then $\delta_1(f)$ is optimal.

\begin{proposition}\label{prop:bound-delta1-by-dense-faces}
	Let $X$ be a $d$-dimensional $\lambda$-one-sided local spectral expander and $0 < \eta < 1$ a density constant.
	For any $k$-cochain $f \in C^k(X)$, $1 \le k < d$,
	$$\norm{\delta_1(f)} \ge
	(k+2)\norm{f}\left(1 - (k+1)\Big(\lambda+\eta+\frac{\norm{\dense_{k-1}}}{\norm{f}}\Big)\right).$$
\end{proposition}

The proof of this proposition will follow from the following two lemmas. The first lemma holds for any simplicial complex.
\begin{lemma}\label{lem:delta1-composition-to-links}
	Let $X$ be a $d$-dimensional simplicial complex. 
	For any $k$-cochain $f \in C^k(X)$, $1 \le k < d$,
	$$\norm{\delta_1(f)} \ge
	(k+2)\bigg(\frac{1}{2}\sum_{\sigma \in X(k-1)}\norm{(\delta_1(f_\sigma), \sigma)} -
	k\sum_{\sigma \in X(k-1)}\norm{(\delta_2(f_\sigma), \sigma)}
	\bigg)$$
\end{lemma}
\begin{proof}
	Denote by $\delta_i(f)$ the set of $(k+1)$-faces that contain exactly $i$ $k$-faces from $f$. Summing $\delta_1(f_\sigma)$ in the links of all $\sigma \in X(k-1)$ equals
	\begin{equation}\label{lem:delta1-composition-to-links-eq1}
		\sum_{\sigma \in X(k-1)}\norm{(\delta_1(f_\sigma), \sigma)} =
		\sum_{i=1}^{k+1}\frac{i(k+2-i)}{\binom{k+2}{2}}\norm{\delta_i(f)}.
	\end{equation}

	Summing $\delta_2(f_\sigma)$ in the links of all $\sigma \in X(k-1)$ equals
	\begin{equation}\label{lem:delta1-composition-to-links-eq2}
		\sum_{\sigma \in X(k-1)}\norm{(\delta_2(f_\sigma), \sigma)} = 
		\sum_{i=2}^{k+2}\frac{\binom{i}{2}}{\binom{k+2}{2}}\norm{\delta_i(f)}.
	\end{equation}
	
	Multiplying~\eqref{lem:delta1-composition-to-links-eq2} by $2k$ yields
	\begin{equation}\label{lem:delta1-composition-to-links-eq3}
		2k\sum_{\sigma \in X(k-1)}\norm{(\delta_2(f_\sigma), \sigma)} = 
		\sum_{i=2}^{k+2}\frac{i(i-1)k}{\binom{k+2}{2}}\norm{\delta_i(f)} \ge
		\sum_{i=2}^{k+2}\frac{i(k+2-i)}{\binom{k+2}{2}}\norm{\delta_i(f)}.
	\end{equation}
	
	Subtracting~\eqref{lem:delta1-composition-to-links-eq3} from~\eqref{lem:delta1-composition-to-links-eq1} yields
	\begin{equation*}
		\sum_{\sigma \in X(k-1)}\norm{(\delta_1(f_\sigma),\sigma)} -
		2k\sum_{\sigma \in X(k-1)}\norm{(\delta_2(f_\sigma), \sigma)} \le
		\frac{2}{k+2}\norm{\delta_1(f)}.
	\end{equation*}

	Multiplying both sides by $(k+2)/2$ finishes the proof.
%	For any $1 \le i \le k+1$ it holds that
%	\newpage
%	\begin{equation*}
%	\begin{aligned}
%		\norm{\delta_i(f)} &=
%		\Pr_{\sigma \in X(k+1)}[\sigma \in \delta_i(f)]\cdot
%		\frac{\Pr_{\tau \in X(k-1)}[\sigma\setminus\tau \in  \delta_1(f_\tau) \;|\; \tau \subset \sigma]}{\Pr_{\tau \in X(k-1)}[\sigma\setminus\tau \in  \delta_1(f_\tau) \;|\; \tau \subset \sigma]} \\&=
%		\Pr_{\sigma \in X(k+1)}[\sigma \in \delta_i(f)]\cdot\Pr_{\tau \in X(k-1)}[\sigma\setminus\tau \in  \delta_1(f_\tau) \;|\; \tau \subset \sigma]\cdot
%		\frac{\binom{k+2}{2}}{i(k+2-i)} \\&=
%		\Pr_{\substack{\sigma \in X(k+1)\\\tau \in X(k-1)}}[]
%	\end{aligned}
%	\end{equation*}
\end{proof}

The following lemma holds for any $\lambda$-one-sided local spectral expander.
\begin{lemma}\label{lem:bound-delta1-delta2-in-links}
	Let $X$ be a $d$-dimensional $\lambda$-one-sided local spectral expander and $0<\eta<1$ a density constant. For any $k$-cochain $f \in C^k(X)$, $1 \le k < d$,
	\begin{enumerate}[label=(\arabic*)]
		\item $\displaystyle \sum_{\sigma \in X(k-1)}\norm{(\delta_1(f_\sigma), \sigma)} \ge
		2(1-\lambda-\eta)\norm{(f, \sparse_{k-1})},$
		
		\item $\displaystyle\sum_{\sigma \in X(k-1)}\norm{(\delta_2(f_\sigma), \sigma)} \le
		\norm{(f, \dense_{k-1})} + (\lambda + \eta)\norm{(f, \sparse_{k-1})},$
	\end{enumerate}
	where $\sparse_{k-1} = X(k-1)\setminus \dense_{k-1}$.
\end{lemma}
\begin{proof}
	Since $X$ is a one-sided local spectral expander, $f_\sigma$ is a subset of vertices in $X_\sigma$ so both inequalities follow immediately form the known Cheeger inequality.
\end{proof}

We can now prove proposition~\ref{prop:bound-delta1-by-dense-faces}.
\begin{proof}[Proof of proposition~\ref{prop:bound-delta1-by-dense-faces}]
Since
\begin{equation*}
	\norm{f} =
	\norm{(f, \dense_{k-1})} + \norm{(f, \sparse_{k-1})},
\end{equation*}
lemma~\ref{lem:bound-delta1-delta2-in-links}\textit{(1)} yields
\begin{equation}\label{prop:bound-delta1-by-dense-faces-eq1}
	\sum_{\sigma \in X(k-1)}\norm{(\delta_1(f_\sigma), \sigma)} \ge
	2(1-\lambda-\eta)\norm{f} - 2\norm{(f, \dense_{k-1})},
\end{equation}
and lemma~\ref{lem:bound-delta1-delta2-in-links}\textit{(2)} yields
\begin{equation}\label{prop:bound-delta1-by-dense-faces-eq2}
	\sum_{\sigma \in X(k-1)}\norm{(\delta_2(f_\sigma), \sigma)} \le
	(\lambda + \eta)\norm{f} + \norm{(f, \dense_{k-1})}.
\end{equation}

Substituting~\eqref{prop:bound-delta1-by-dense-faces-eq1} and~\eqref{prop:bound-delta1-by-dense-faces-eq2} in lemma~\ref{lem:delta1-composition-to-links} finishes the proof.
\end{proof}

\subsection{Part II - Bound the fraction of dense $(k-1)$-faces.}

We show in this section that for every double balanced and small cochain in a good enough one-sided local spectral expander, the fraction of dense $(k-1)$-faces is very small.

We first extend the definition of dense faces to every dimension $-1 \le i \le k-1$. Given a density constant $0 < \eta < 1$ and $\varepsilon > 0$, we set $\eta_{k-1} = \eta$ and for every $0 \le i \le k-1$ we define
$$\eta_{i-1} = \frac{\eta_i}{\alpha} - \frac{\varepsilon}{(k+1)^2\alpha^{k-i}}.$$
We then define the dense faces in dimension $i$ to be
%. Let $\eta_{k-1}, \eta_{k-2}, \dotsc, \eta_{-1}$ be denseness constants. For every $-1 \le i < k$, we define
$$\dense_i = \{\sigma \in X(i) \;|\; \norm{f_\sigma} > \eta_i\}.$$

Our goal in this subsection is to prove the following proposition.
\begin{proposition}\label{prop:bound-dense-faces}
	Let $X$ be a $d$-dimensional $\lambda$-one-sided local spectral expander, $1 \le k < d$ any dimension, $\alpha \ge 1$ a balance constant, $0 < \eta < 1$ a density constant and $\varepsilon > 0$. For any $k$-cochain $f \in C^k(X)$ such that $f$ is $\alpha$-double balanced and $\displaystyle \norm{f} \le \eta_{-1}$ it holds that
	$$\norm{\dense_{k-1}} \le 3k!\bigg(\frac{(k+1)^3\alpha^{k}\lambda}{\varepsilon}\bigg)^2\norm{f}$$
\end{proposition}

%\begin{corollary}
%	For any $0< \theta < 1$,
%	if $\displaystyle \lambda^2 \le \frac{\eta^2\theta}{k(k+1)((k+1)!)^3\alpha^{2k}}$, then
%	$$\norm{\dense_{k-1}} \le \theta\norm{f}.$$
%\end{corollary}

We start by showing that in a $\lambda$-one-sided local spectral expander, the restriction of a cochain to almost every vertex is seen with the right proportion.
\begin{lemma}\label{lem:in-front-seen-with-right-fraction}
	Let $X$ be a $d$-dimensional $\lambda$-one-sided local spectral expander. For any $k$-cochain $f \in C^k(X)$, $0 \le k < d$, and $\varepsilon > 0$ it holds that
	$$\Pr_{u\in X(0)}[\norm{f^u} > \norm{f} + \varepsilon] \le \left(\frac{(k+1)\lambda}{\varepsilon}\right)^2\norm{f}.$$
\end{lemma}
\begin{proof}
	Define the following graph $G = (V,E)$, where $V = X(k)$, i.e., all $k$-faces of $X$, and 
	$E = \big\{\{\sigma_1,\sigma_2\} \;|\; \exists u \in X(0) \mbox{ s.t. } \sigma_1 \cupdot u, \sigma_2\cupdot u \in X(k+1) \big\}$, i.e., there is an edge between $\sigma_1$ and $\sigma_2$ if and only if there exists some vertex in $X$ that completes both $\sigma_1$ and $\sigma_2$ to a $(k+1)$-face.
	
	We define a probability distribution on $G$ that corresponds to the probability distribution of $X$ as follows:
	\begin{itemize}
		\item The probability of a vertex $\sigma \in V$ equals to the probability of the corresponding $k$-face $\sigma \in X(k)$.
		\item The probability of an edge $\{\sigma_1,\sigma_2\} \in E$ equals
		$\E_{u \in X(0)}\Pr[\sigma_1 \cupdot u \;|\; u]\cdot\Pr[\sigma_2 \cupdot u \;|\; u]$,
		where all the probabilities are taken according to the complex $X$.
	\end{itemize}
	
	Since $X$ is a $\lambda$-one-sided local spectral expander, by~\cite[Claim~4.9]{DD19} $G$ is a $((k+1)\lambda)^2$-spectral expander, because its adjacency operator is a two steps walk of the $0,2$-complement walk of~\cite{DD19}.
	
	Now, define $\mu : X(0) \to \R$ by $\mu(u) = \norm{f^u} = \Pr[\sigma \in f \;|\; \sigma\cupdot u \in X(k+1)]$. The following holds by laws of probability:
	\begin{equation}\label{lem:in-front-seen-with-right-fraction-eq1}
		\E_{u \in X(0)}[\mu(u)] = \E_{u \in X(0)}\Pr[\sigma \in f \;|\; \sigma\cupdot u \in X(k+1)] = \Pr[\sigma\in f] = \norm{f}.
	\end{equation}
	
	\begin{equation}\label{lem:in-front-seen-with-right-fraction-eq2}
		\begin{aligned}
			\E_{u \in X(0)}[\mu(u)^2] &= \E_{u \in X(0)}\Pr[\sigma_1 \in f \;|\; \sigma_1\cupdot u \in X(k+1)]\cdot\Pr[\sigma_2 \in f \;|\; \sigma_2\cupdot u \in X(k+1)] \\[5pt]&=
			\Pr_{\{\sigma_1,\sigma_2\} \in E}[\sigma_1 \in f \wedge \sigma_2 \in f] =
			\norm{E(f)},
		\end{aligned}
	\end{equation}
	where $E(f)$ is the set of edges $\{\sigma_1,\sigma_2\}$ in $G$ such that both $\sigma_1$ and $\sigma_2$ are in $f$. Since $G$ is a $((k+1)\lambda)^2$-spectral expander, it follows that
	$\norm{E(A)} \le \norm{f}^2 + ((k+1)\lambda)^2\norm{A}$. Substituting in~\eqref{lem:in-front-seen-with-right-fraction-eq2} and combining~\eqref{lem:in-front-seen-with-right-fraction-eq1} yields
	\begin{equation*}
		\Var_{u\in X(0)}[\mu(u)] = \E_{u \in X(0)}[\mu(u)^2] - \E_{u \in X(0)}[\mu(u)]^2 \le ((k+1)\lambda)^2\norm{f}.
	\end{equation*}
	
	Now, by Chebyshev's inequality
	\begin{align*}
		\Pr\big[\norm{f^u} > \norm{f} + \varepsilon\big] =
		\Pr\big[\mu(u) > \E[\mu] + \varepsilon\big] \le
		\frac{\Var[\mu]}{\varepsilon^2} \le
		\bigg(\frac{(k+1)\lambda}{\varepsilon}\bigg)^2\norm{f}.
	\end{align*}
	This completes the proof.
\end{proof}

In the next lemma we show that for every dimension $i$, if $f$ is double balanced in dimension $i$ then the fraction of dense $i$-faces is not much more than the fraction of dense $(i-1)$-faces.
\begin{lemma}\label{lem:dense-faces-one-dimension-lower}
	Let $X$ be a $d$-dimensional $\lambda$-one-sided local spectral expander and $f \in C^k(X)$, $0 \le k < d$. For every $0 \le i < k$, if $f$ is $\alpha$-double balanced in dimension $i$ then
	$$\norm{\dense_i} \le
	(i+1)\norm{\dense_{i-1}} + (i+1)
	\left(\frac{(k+1-i)(k+1)^2\alpha^{k-i}\lambda}{\varepsilon}\right)^2\norm{f}.$$
\end{lemma}

\begin{proof}
	Note that for every $\sigma \in \dense_i$ there must exist a vertex $u \in \sigma$ such that
	\begin{equation}\label{lem:dense-faces-one-dimension-lower-eq1}
		\norm{(f_{\sigma\setminus u})^u} > \frac{\eta_i}{\alpha},
	\end{equation}
	since otherwise
	$$\norm{f_\sigma} \le \frac{\alpha}{i+1}\sum_{u \in \sigma}\norm{(f_{\sigma\setminus u})^u} \le \eta_i$$
	and $\sigma \notin \dense_i$.
	
	For every $\sigma \in \dense_i$, fix one $(i-1)$-face $\tau(\sigma) = \sigma \setminus u$ that satisfies~\eqref{lem:dense-faces-one-dimension-lower-eq1}. By laws of probability
	\begin{equation}\label{lem:dense-faces-one-dimension-lower-eq2}
		\begin{aligned}
			\norm{\dense_i} &=
			\Pr[\sigma_i \in \dense_i] =
			(i+1)\Pr[\sigma_i \in \dense_i \wedge \sigma_{i-1} = 	\tau(\sigma_i)] \\[6pt]&\le
			(i+1)\norm{\dense_{i-1}} + (i+1)\Pr[\sigma_i \in \dense_i \wedge \sigma_{i-1} = 	\tau(\sigma_i) \;|\; \tau(\sigma_i) \notin \dense_{i-1}],
		\end{aligned}
	\end{equation}
	where the inequality holds by splitting to the two cases whether $\tau(\sigma_i) \in \dense_{i-1}$.
	
	We focus now on the right summand of~\eqref{lem:dense-faces-one-dimension-lower-eq2} which is the case where $\tau(\sigma_i) \notin \dense_{i-1}$. Recall that $\tau(\sigma_i)$ satisfies~\eqref{lem:dense-faces-one-dimension-lower-eq1}. Thus, we can bound the probability of this event by the probability to choose a sparse $(i-1)$-face and then a vertex such that~\eqref{lem:dense-faces-one-dimension-lower-eq1} holds, i.e.,
	\begin{equation}\label{lem:dense-faces-one-dimension-lower-eq3}
		\Pr[\sigma_i \in \dense_i \wedge \sigma_{i-1} = 	\tau(\sigma_i) \;|\; \tau(\sigma_i) \notin \dense_{i-1}] \le
		\E_{\tau \in \subsparse_{i-1}}\Pr_{u \in X_\tau(0)}
		\left[\norm{(f_\tau)^u} > \frac{\eta_i}{\alpha}\right].
	\end{equation}
	
	Since $\tau \in \sparse_{i-1}$, it holds that $\norm{f_\tau} \le \eta_{i-1}$. Thus,
	\begin{equation}\label{lem:dense-faces-one-dimension-lower-eq4}
		\E_{\tau \in \subsparse_{i-1}}\Pr_{u \in X_\tau(0)}
		\left[\norm{(f_\tau)^u} > \frac{\eta_i}{\alpha}\right] \le
		\E_{\tau \in \subsparse_{i-1}}\Pr_{u \in X_\tau(0)}
		\left[\norm{(f_\tau)^u} > \norm{f_\tau} + \frac{\varepsilon}{(k+1)^2\alpha^{k-i}} \right],
	\end{equation}
	where the inequality holds since
	$$\norm{f_\tau} + \frac{\varepsilon}{(k+1)^2\alpha^{k-i}} \le \eta_{i-1} + \frac{\varepsilon}{(k+1)^2\alpha^{k-i}} = \frac{\eta_i}{\alpha}.$$
Combining~\eqref{lem:dense-faces-one-dimension-lower-eq2},~\eqref{lem:dense-faces-one-dimension-lower-eq3} and~\eqref{lem:dense-faces-one-dimension-lower-eq4} yields
	\begin{equation*}
		\begin{aligned}
			\norm{\dense_i} &\le
			(i+1)\norm{\dense_{i-1}} + (i+1)\E_{\tau \in \subsparse_{i-1}}\Pr_{u \in X_\tau(0)}
			\left[\norm{(f_\tau)^u} > \norm{f_\tau} + \frac{\varepsilon}{(k+1)^2\alpha^{k-i}} \right] \\&\le
			(i+1)\norm{\dense_{i-1}} + (i+1)\E_{\tau \in \subsparse_{i-1}}\left[\bigg(\frac{(k+1-i)(k+1)^2\alpha^{k-i}\lambda}{\varepsilon}\bigg)^2\norm{f_\tau}\right] \\&\le
			(i+1)\norm{\dense_{i-1}} + (i+1)\bigg(\frac{(k+1-i)(k+1)^2\alpha^{k-i}\lambda}{\varepsilon}\bigg)^2\norm{f},
		\end{aligned}
	\end{equation*}
	where the second inequality follows by lemma~\ref{lem:in-front-seen-with-right-fraction}. This completes the proof.
\end{proof}

We can now prove proposition~\ref{prop:bound-dense-faces}.
\begin{proof}[Proof of proposition~\ref{prop:bound-dense-faces}]
	We apply lemma~\ref{lem:dense-faces-one-dimension-lower} for $i = k-1,k-2,\dotsc,0$ step by step.
	\begin{equation*}
		\begin{aligned}
			\norm{\dense_{k-1}} &\le
			k\norm{\dense_{k-2}} + k\bigg(\frac{2(k+1)^2\alpha\lambda}{\varepsilon}\bigg)^2\norm{f} \\&\le
			k(k-1)\norm{\dense_{k-3}} + \left(k(k-1)\bigg(\frac{3(k+1)^2\alpha^2\lambda}{\varepsilon}\bigg)^2 + k\bigg(\frac{2(k+1)^2\alpha\lambda}{\varepsilon}\bigg)^2\right)\norm{f} \le \\\dotsb &\le
			k!\norm{\dense_{-1}} + \left(k!\bigg(\frac{(k+1)^3\alpha^k\lambda}{\varepsilon}\bigg)^2 +\dotsb + k\bigg(\frac{2(k+1)^2\alpha\lambda}{\varepsilon}\bigg)^2\right)\norm{f} \\&=
			\sum_{i=0}^{k-1}\frac{k!}{i!}\bigg(\frac{(k+1-i)(k+1)^2\alpha^{k-i}\lambda}{\varepsilon}\bigg)^2\norm{f} \\&\le
			k!\bigg(\frac{(k+1)^2\alpha^k\lambda}{\varepsilon}\bigg)^{\!2}\;\sum_{i=0}^{k-1}\frac{(k+1-i)^2}{i!}\norm{f} \\&\le
			3k!\bigg(\frac{(k+1)^3\alpha^k\lambda}{\varepsilon}\bigg)^{\!2}\norm{f}
		\end{aligned}
	\end{equation*}
	where the equality holds since $\displaystyle \norm{f_\emptyset} = \norm{f} \le \eta_{-1}$, i.e., the empty set is not dense, and hence $\norm{\dense_{-1}} = 0$. The rest of the inequalities are just calculations. This completes the proof.
\end{proof}

\subsection{Proof of theorems~\ref{thm:optimal-delta1-for-small-sets} and~\ref{thm:some-delta1-for-small-sets}}

\begin{proof}[Proof of theorem~\ref{thm:optimal-delta1-for-small-sets}]
	Let $\displaystyle \lambda \le \frac{\varepsilon}{d^3\alpha^{d-1}}\sqrt{\frac{\varepsilon}{3d!}}$ and $\displaystyle \eta = \frac{\varepsilon}{(k+1)}$. By simple calculation $$\eta_{-1} = \frac{\varepsilon}{(k+1)^2\alpha^k}.$$
	 Thus, since $\norm{f} \le \eta_{-1}$, proposition~\ref{prop:bound-dense-faces} implies that
	\begin{equation}\label{thm:optimal-delta1-for-small-sets-eq1}
		\norm{\dense_{k-1}} \le 3k!\left(\frac{(k+1)^3\alpha^k\lambda}{\varepsilon}\right)^2
		\le \frac{\varepsilon}{k+1}\norm{f}.
	\end{equation}
	
	Substituting~\eqref{thm:optimal-delta1-for-small-sets-eq1} in proposition~\ref{prop:bound-delta1-by-dense-faces} finishes the proof.
\end{proof}

\begin{proof}[Proof of theorem~\ref{thm:some-delta1-for-small-sets}]
	Let $\displaystyle \lambda \le \frac{\varepsilon}{d^3\alpha^{d-1}}\sqrt{\frac{\varepsilon}{(d+1)!}}$ and $\displaystyle \eta = \frac{1-\varepsilon/(k+1)}{(k+1)}$. By simple calculation $$\eta_{-1} = \frac{1-\varepsilon}{(k+1)\alpha^k}.$$
	 Thus, since $\norm{f} \le \eta_{-1}$, proposition~\ref{prop:bound-dense-faces} implies that
	\begin{equation}\label{thm:optimal-delta1-for-small-sets-eq1}
		\norm{\dense_{k-1}} \le 3k!\left(\frac{(k+1)^3\alpha^k\lambda}{\varepsilon}\right)^2
		\le \frac{\varepsilon}{(k+2)(k+1)}\norm{f}.
	\end{equation}
	
	Substituting~\eqref{thm:optimal-delta1-for-small-sets-eq1} in proposition~\ref{prop:bound-delta1-by-dense-faces} finishes the proof.
\end{proof}

We conclude this section by noting that the proof of lemma~\ref{intro-lem:relation-to-common-def} from the introduction is exactly the same as the proof of proposition~\ref{prop:bound-dense-faces}, with the only difference that we start by setting $\eta_\ell = \varepsilon$ and bound the fraction of dense $\ell$-faces instead of the dense $(k-1)$-faces.
%\begin{lemma}\label{lem:relation-to-common-def}
%For every $d \ge 2$, $\alpha \ge 1$ and $0 < \varepsilon < 1$ there exists $\lambda = \lambda(d, \alpha, \varepsilon)$ such that the following holds:
%	Let $X$ be a $d$-dimensional $\lambda$-local spectral expander. For any $k$-cochain $f \in C^k(X)$, $1 \le k < d$, and any dimension $\ell < k$, if $f$ is $\alpha$-double balanced and $\displaystyle \norm{f} \le \frac{\varepsilon}{(\ell+1)\alpha^\ell}$ then
%$$\Pr_{\sigma \in X(\ell)}[\norm{f_\sigma} \le \varepsilon] \ge (1-\varepsilon)\norm{f}.$$
%\end{lemma}

\section{Cohomologies are double balanced}
Previous works could only obtain complexes with some constant lower bound on the size of their cohomologies~\cite{KKL14,EK16,KM21}. We show that for high dimensional expanders (in a topological sense), all of their cohomology elements are double balanced. We then utilize the $\delta_1$-expansion of double balanced sets in order to obtain a lower bound on their size, achieving an \emph{exponential} improvement upon the current state of the art.

We start by proving theorem~\ref{intro-thm:cohomologies-double-balanced} from the introduction, which we restate here in a formal way.

\begin{theorem}[Cohomologies are double balanced]\label{thm:cohomologies-double-balanced}
Let $X$ be a $d$-dimensional complex such that every non-trivial link in $X$ is a $\beta$-coboundary expander. For every $\ell < k < d$, any $k$-cohomology element is $\displaystyle\frac{\ell+1}{\beta}$-double balanced in dimension $\ell$.
\end{theorem}

\begin{proof}
Let $f \in H^k(X)$ be a $k$-cohomology and $\sigma \in X(\ell)$ be an $\ell$-face. Consider a $(k-\ell)$-face $\tau \in \delta(f_\sigma)$. Let us denote $\sigma = \{v_0,v_1,\dotsc,v_\ell\}$ and $\tau = \{v_{\ell+1},v_{\ell+2},\dotsc,v_{k+1}\}$. By definition
$$\sum_{i=\ell+1}^{k+1}f(\sigma \cup \tau \setminus v_i) =
\sum_{i=\ell+1}^{k+1}f_\sigma(\tau \setminus v_{i}) \ne 0,$$
where the inequality holds since $\tau \in \delta(f_\sigma)$.
Since $f$ is a $k$-cohomology, it holds that
$$\sum_{i=0}^{k+1}f(\sigma \cup \tau \setminus v_i) = 0.$$
Therefore, there must exist $0 \le j \le \ell$ such that $f(\sigma \cup \tau \setminus v_j) \ne 0$. By definition of restriction and localization, it means that
$$(f_{\sigma \setminus v_j})^{v_j}(\tau) =
(f_{\sigma \setminus v_j})(\tau) \ne 0.$$
In other words, for every $\tau \in \delta(f_\sigma)$, there exists a vertex $v \in \sigma$ such that $\tau \in (f_{\sigma \setminus v})^v$. It follows that
\begin{equation}\label{thm:cohomologies-balanced-eq1}
\norm{\delta(f_\sigma)} \le \sum_{v \in \sigma}\norm{(f_{\sigma \setminus v})^v}.
\end{equation}
Now, since $f$ is a $k$-cohomology, $f$ is minimal and hence also locally minimal. The $\beta$-coboundary expansion of the links implies that
\begin{equation}\label{thm:cohomologies-balanced-eq2}
\norm{\delta(f_\sigma)} \ge \beta\norm{f_\sigma}.
\end{equation}

Combining~\eqref{thm:cohomologies-balanced-eq1} and~\eqref{thm:cohomologies-balanced-eq2} implies that
$$\norm{f_\sigma} \le
\frac{1}{\beta}\sum_{v \in \sigma}\norm{(f_{\sigma \setminus v})^v} =
\frac{\ell+1}{\beta}\E_{v \in \sigma}\norm{(f_{\sigma \setminus v})^v}.$$
This complete the proof.
\end{proof}

We conclude by proving theorem~\ref{intro-thm:lower-bound-cohomologies} from the introduction, which we restate here in a formal way.

\begin{theorem}[Lower bound on cohomology elements]
For every $d \ge 2$, $\beta > 0$ and $\varepsilon > 0$ there exists $\lambda = \lambda(d, \beta, \varepsilon)$ such that the following holds. Let $X$ be a $d$-dimensional $\lambda$-one-sided local spectral expander such that every non-trivial link in $X$ is a $\beta$-coboundary expander. For every $k < d$, any $k$-cohomology element $f \in H^k(X)$ satisfies
$$\norm{f} \ge \frac{(1-\varepsilon)\beta^k}{(k+1)!}.$$
\end{theorem}
\begin{proof}
Assume towards contradiction that there exists $f \in H^k(X)$ with $\displaystyle\norm{f} < \frac{(1-\varepsilon)\beta^k}{(k+1)!}$. By theorem~\ref{thm:cohomologies-double-balanced}, $f$ is $\big((\ell+1)/\beta\big)$-double balanced in dimension $\ell$ for every $\ell < k$. Then theorem~\ref{thm:some-delta1-for-small-sets} implies\footnote{It is implied by the proof of theorem~\ref{thm:some-delta1-for-small-sets} if we consider for every dimension $\ell$ its own double balance constant $(\ell+1)/\beta$ rather than bounding all constants by the largest one.} that $\norm{\delta_1(f)} > 0$ in contradiction to $f$ being a cohomology elements (i.e., $\delta(f) = 0$).
\end{proof}

\bibliography{Bibliography}

\end{document}